\documentclass{article}

\usepackage{times}

\usepackage{natbib}

\usepackage{enumitem}
\usepackage{booktabs}
\usepackage{latexsym}
\usepackage{amssymb}
\usepackage{amsmath}
\usepackage{amsthm}
\usepackage{stmaryrd}
\usepackage{flushend} 
\usepackage{bbm}

\usepackage{verbatim}
\usepackage[lined,plain,commentsnumbered]{algorithm2e}
\usepackage{tikz}\usetikzlibrary{matrix,arrows}

\newtheorem{theorem}{Theorem}
\newtheorem{proposition}[theorem]{Proposition}
\newtheorem{lemma}[theorem]{Lemma}
\newtheorem{corollary}[theorem]{Corollary}

\newtheorem{example}{Example}
\newtheorem{definition}{Definition}

\newcommand{\avg}{\operatornamewithlimits{avg}}

\newcommand{\rev}{{\bf{r}}}

\renewcommand{\P}{\mathbb{P}}
\renewcommand{\O}{\mathbb{O}}

\newcommand{\putaway}[1]{}

\newcommand{\orat}{\textsc{$\mathbbmss{O}$-$\mathbbmss{rating}$}}
\newcommand{\prat}{\textsc{$\mathbbmss{P}$-$\mathbbmss{rating}$}}
\renewcommand{\deg}{\text{\it deg}}
\newcommand{\eval}{\text{\it eval}}

 \usepackage{authblk} 
\usepackage [english]{babel}
\usepackage [autostyle, english = american]{csquotes}
\MakeOuterQuote{"}

\begin{document}
\title{A network-based rating system \\ and its resistance to bribery}
\author[1]{Umberto Grandi}
\author[2]{Paolo Turrini}

\affil[1]{IRIT, University of Toulouse} 
\affil[2]{Department of Computing, Imperial College London}


\date{} 

\maketitle

\begin{abstract}

We study a rating system in which a set of individuals (e.g., the customers  of a restaurant) evaluate a given service (e.g, the restaurant), with their aggregated opinion determining the probability of all individuals to use the service and thus its generated revenue. We explicitly model the influence relation by a social network,  with individuals being influenced by the evaluation of their trusted peers. 
On top of that we allow a malicious service provider (e.g., the restaurant owner) to bribe some individuals, i.e., to invest a part of his or her expected income to modify their opinion, therefore influencing his or her final gain. 
We analyse the effect of bribing strategies under various constraints, and we show under what conditions the system is bribery-proof, i.e., no bribing strategy yields a strictly positive expected gain to the service provider.


\end{abstract}

\section{Introduction}

Imagine to be the owner of a new and still relatively unknown restaurant.
The quality of food is not spectacular and the customers you have seen so far are only limited to a tiny number of friends of yours.
Your account on Tripadvisor$^{\tiny \textregistered}$ has received no review and your financial prospects look grim at best.
There is one easy solution to your problems: you ask your friends to write an enthusiastic review for you, in exchange for a free meal.
After this, Tripadvisor$^{\tiny \textregistered}$ lists your restaurant as excellent and the number of customers, together with your profit, suddenly florishes.

Systems such as Tripadvisor$^{\tiny \textregistered}$, where a small proportion of customers writes reviews and influences a large number of potential customers, are not {\em bribery-proof}: each restaurant owner - or the owner of whichever service - is able to offer a compensation - monetary or not - in exchange for positive evaluation, having an impact on the whole set of potential customers. Tripadvisor$^{\tiny \textregistered}$ is based on what we call "Objective Rating", or $\orat$: individual evaluations are aggregated into a single figure, which is seen by, and thus influences, every potential customer.

What we study in this paper is a system in which each individual {\em only} receives the evaluation given by the set of trusted peers, his or her friends, and only this aggregated opinion influences his or her decision. This is what we call "Personalised Rating", or $\prat$, which can be seen a generalisation of $\orat$ in which influence has a complex network-structure.
So, while in the case of $\orat$ the restaurant owner knows exactly how influence flows among the customers, this might not be the case with $\prat$.

{{\bf Our contribution}} We analyse the effect of bribing strategies in the case of $\orat$ and $\prat$ under various constraints, depending on the presence of customers who do not express any opinion and the knowledge of the network by the service provider: the exact network is known, the network is known but not the customers' exact position, the network is completely unknown. 
We show under what conditions the system is bribery-proof, i.e., there is no bribe yielding a strictly positive expected gain to the  service provider, and we provide algorithms for the computation of (all) optimal bribing strategies when they exist.

Intuitively, being able to know and bribe influential customers is crucial for guaranteeing a positive expected reward of a bribing strategy. However, while with large populations of non-voters "random" bribes can still be profitable, the effect of $\prat$ is largely different from that of $\orat$ and, as we show, the expected profit in the former can be severely limited and drops below zero in all networks, under certain (mild) conditions on the cost of bribes.

Our study can be applied to all situations in which individuals influence one another in the opinion they give and bribery can have a disruptive role in determining collective decisions.

{{\bf Related research lines}} Our approach relates to several research lines in artificial intelligence, game theory and (computational) social choice \cite{HandbookCOMSOC2015}.

\begin{description}

\item[Network-based voting and mechanism design] 
We study social networks in which individuals' local decisions can be manipulated 
to modify the resulting global properties. A similar approach is taken by \citet{DBLP:journals/fuin/AptM14} and \citet{DBLP:journals/logcom/SimonA15}, which study the changes on a social network needed to make a certain product adopted among users.  
Further contributions include rational secret sharing and multi-party computation \cite{abraham06distributed}, the strategic manipulation of peer reviews \cite{DBLP:conf/ijcai/KurokawaLMP15}, and the growing literature on voting in social networks \cite{ConitzerMSS2012,Salehi-AbariAAMAS2014,ElkindWINE2014,TsangEtAlAAMAS2015,ProcacciaEtAlIJCAI2015}.

\item[Lobbying and Bribery] Our framework features an external agent trying to influence individual decisions to reach his or her private objectives. Lobbying in decision-making is an important problem in the area of social choice, from the seminal contribution of \cite{NBERw6589} to more recent studies in multi-issue voting \cite{ChristianEtAl2007}. 
Lobbying and bribery are also established concepts in computational social choice, with their computational complexity being analysed extensively
\cite{FaliszewskiEtAlJAIR2009,DBLP:conf/aldt/BaumeisterER11,BredereckJAIR2014,BredereckEtAlJAIR2016}.

\item[Reputation-based systems] We study the aggregation of possibly insincere individual evaluations by agents that can influence one another through trust relations.
In this sense ours can be seen as a study of reputation in Multi Agent Systems, which has been an important concern of MAS for the past decades \cite{conte-paolucci,SabaterEtAlAIR2005,GarcinEtAl2009}. 
In particular, our framework treats reputation as a manipulable piece of information, not just a static aggregate of individual opinions, coherently with the work of \citet{DBLP:journals/advcs/ContePS08} and \citet{DBLP:journals/air/PinyolS13}.

\end{description}

{{\bf Paper structure}} Section~\ref{sec:basic} presents the basic setup,  introducing $\orat$, $\prat$ and bribing strategies. Section~\ref{sec:objective} focusses on $\orat$, studying its bribery-proofness under various knowledge conditions. 
Section~\ref{sec:personalised} evaluates $\prat$ against the same knowledge conditions. In Section \ref{sec:strategyproof} we compare the two systems, taking the cost of bribery into account.
We conclude by summarising the main findings and pointing at future research directions (Section~\ref{sec:conclusions}).

\section{Basic setup}\label{sec:basic}

In this section we provide the basic formal definitions.

\subsection{Restaurant and customers}

Our framework features an object $r$, called {\em restaurant}, being evaluated by a finite non-empty set of individuals $C=\{c_1, \ldots, c_n\}$, called \emph{customers}.  Customers are connected by an undirected graph 
$E \subseteq C \times C$, called the \emph{customers network}. Given $c\in C$ we call $N(c)=\{x\in C \mid (c,x)\in E\}$ the {\em neighbourhood} of $c$, always including $c$ itself.


Customers concurrently submit an \emph{evaluation} of the restaurant, drawn from a set of values $\text{\it Val}\subseteq [0,1]$, together with a distinguished element $\{*\}$, symbolising no opinion. 
Examples of values are the set~$[0,1]$ itself, or a discrete assignment of 1 to 5 stars, as common in online rating systems.
We make the assumption that $\{0,1\}\subseteq \text{\it Val}$ and that $\text{\it Val}$ is closed under the operation $\min\{1, x+y\}$ for all $x,y\in \text{\it Val}$. 
The vast majority of known rating methods can be mapped onto the $[0,1]$ interval and analysed within our framework.

We represent the evaluation of the customers as a function $\mathit{eval}: C \to \text{\it Val}\cup \{*\}$ and define $V \subseteq C$ as the subset of customers that expresses an evaluation over the restaurant, i.e., $V=\{ c\in C \mid \eval(c)\not = *\}$. We refer to this set as the set of \emph{voters} and we assume it to be always non-empty, i.e., there is at least one customer that expresses an evaluation.
%

\subsection{Two rating systems}

In online rating systems such as Tripadvisor$^{\tiny \textregistered}$ every interested customer can see - and is therefore influenced by - (the average of) what the other customers have written. We call this method $\orat$, which stands for \emph{objective rating}. 

Given an evaluation function $\eval$ of a restaurant, the associated $\orat$ is defined as follows:

$$ \orat(\eval) = {\avg_{c\in V}} \, \eval (c)  $$

\noindent 
Where $\avg$ is the average function across real-valued $\mathit{eval}(c)$, disregarding $*$. We omit $eval$ when clear from the context.

$\orat$ flattens individual evaluations into a unique objective aggregate, the rating that a certain restaurant is given. 
What we propose is a refinement of $\orat$, which takes the network of influence into account. In this system customers are {\em only} interested in the evaluation of other customers they can trust, e.g., their friends.
We call our method $\prat$, which stands for \emph{personalised rating}. It is defined for a pair customer-evaluation $(c,\eval)$ as follows:

$$\prat(c,\eval)= {\avg_{k\in N(c)\cap V} \eval(k) } $$

\noindent
So the $\prat(c,eval)$ calculates what customer $c$ comes to think of the restaurant, taking the average of the opinions of the customers $c$ is connected to.
Again we omit $eval$ whenever clear from the context.

Observe that in case a customer has no connection with a voter, then $\prat$ is not defined. To facilitate the analysis we make the technical assumption that \emph{each customer is connected to at least one voter}. Also observe that when $E=C\times C$, i.e., in case the network is complete and each individual is influenced by each other individual, then for all $c\in C$ and $\eval$ we have that $\prat(c, \eval)=\orat(\eval)$.


\subsection{Utilities and strategies}

We interpret a customer evaluation as a measure of his or her \emph{propensity} to go to the restaurant. We therefore assume that the utility that a restaurant gets is proportional to its rating. To simplify the analysis 
we assume a factor 1 proportionality.

{\bf The case of $\orat$}.
For the $\orat$, we assume that the initial utility $u^{0}$ of the restaurant is defined as:
$$u^0_\O = |C| \orat (\eval).$$

Intuitively, the initial utility amounts to the number of customers that actually go to the restaurant, weighted with their (average) predisposition.

At the initial stage of the game, the restaurant owner receives $u^0$, and can then decide to invest a part of it to influence a subset of customers and improve upon the initial gain. We assume utility to be fully transferrable and, to facilitate the analysis, that such transfers translate directly into changes of customers' predispositions.

\begin{definition}\label{def:strategy}
A strategy is a function $\sigma: C \to \text{\it Val}$ such that $\sum_{c\in C} \sigma(c) \leq u^0$. 
\end{definition}

\noindent

Definition~\ref{def:strategy} imposes that strategies are \emph{budget balanced}, i.e., restaurants can only pay with resources they have. 

Let $\Sigma$ be the set of all strategies.
We denote $\sigma^{0}$ the strategy that assigns $0$ to all customers and we call \emph{bribing strategy} any strategy that is different from $\sigma^{0}$. After the execution of a bribing strategy, the evaluation is updated as follows:

\begin{definition}
The evaluation $\eval^\sigma(c)$ after execution of $\sigma$ is  $\mathit{eval}^{\sigma}(c) =  \min \{1,\eval(c)+\sigma(c)\}$, where $*+\sigma(c)=\sigma(c)$, if $\sigma(c) \neq 0$, and $*+\sigma(c)=*$, if $\sigma(c) = 0$.
\end{definition}

\noindent
In this definition we are making the assumption that the effect of bribing a non-voter to vote is equivalent to that of bribing a voter that had a 0-level review, as, intuitively, the individual has no associated predisposition to go to the restaurant.

A strategy is called \emph{efficient} if $\sigma(c)+\eval(c)\leq 1$ for all $c\in C$. 
Let $B(\sigma)=\{c\in C\mid \sigma(c)\not =0\}$ be the set of bribed customers. Let $V^\sigma$ be the set of voters after the execution of~$\sigma$. 
Executing $\sigma$ induces the following change in utility:

$$u_\O^{\sigma} = |C|\orat(\eval^{\sigma}) - \sum_{c\in C}\sigma(c).$$

\noindent
Intuitively,  $u_\O^{\sigma}$ is obtained by adding to the initial utility of the restaurant the rating obtained as an effect of the money invested on each individual minus the amount of money spent.

We define the revenue of a strategy $\sigma$ as the marginal utility obtained by executing it:

\begin{definition}
Let $\sigma$ be a strategy.
The \emph{revenue} of $\sigma$ is defined as $\rev_\O(\sigma)=u_\O^\sigma-u^0$.
We say that $\sigma$ is \emph{profitable} if $\rev_\O(\sigma)>0$.
\end{definition}

\noindent
Finally, we recall the standard notion of dominance:

\begin{definition}
A strategy $\sigma$ is \emph{weakly dominant} if $u_\O^\sigma\geq u_\O^{\sigma'}$ for all $\sigma'{\in} \Sigma$. It is \emph{strictly dominant} if $u_\O^\sigma> u_\O^{\sigma'}$ for all $\sigma{\in} \Sigma$.
\end{definition}

\noindent
Hence a non-profitable strategy is never strictly dominant.

{\bf The case of $\prat$}. The previous definitions can be adapted to the case of $\prat$ as follows:

$$ u^0_\P = \sum_{c\in C} \prat(c,\eval) $$

\noindent
which encodes the initial utility of each restaurant, and

$$u^{\sigma}_\P =\sum_{c\in C}  \prat(c,\eval^\sigma) - \sum_{c\in C}\sigma(c)$$

\noindent
which encodes the utility change after the execution of a $\sigma$.
%
Finally, let the revenue of $\sigma$ be $\rev_\P(\sigma)=u^\sigma_\P - u^0_\P$.
If clear from the context, we use $\prat^\sigma(c)$ for $\prat(\eval^\sigma,c)$.


In order to determine the dominant strategies, we need to establish how the customers vote, how they are connected, and what the restaurant owner knows. In this paper we assume that the restaurant knows $\eval$, leaving the interesting case when $\eval$ is unknown to future work. We focus instead on the following cases: the restaurant knows the network, the restaurant knows the shape of the network but not the individuals' position, and the network is unknown. We analyse the effect of bribing strategies on $\prat$ in each such case. Notice how for the case of $\orat$ the cases collapse to the first. We also look at the special situation in which every customer is a voter. 

Given a set of such assumptions, we say that $\orat$ (or $\prat$) are {\em bribery-proof} under those assumptions if $\sigma^{0}$ is weakly dominant.


{\bf Discussion} 
Our model is built upon a number of simplifying assumptions which do not play a significant role in the results and could therefore be dispensed with:
(i) customers' ratings correspond to their propensity to go to the restaurant. (ii)
the restaurant utility equals the sum of all such propensities (iii) bribe $\sigma(c)$ affects evaluation $\eval(c)$ linearly. 
All these assumptions could be generalised by multiplicative factors, such as an average price $R$ paid at the restaurant, and a "customer price" $D_c$, such that $\eval^\sigma(c)=\eval(c)+\frac{\sigma(c)}{D_c}$.

\section{Bribes under $\orat$}\label{sec:objective}

In this section we look at bribing strategies under $\orat$, first focussing on the case where everyone expresses an opinion, then moving on to the more general case.
\subsection{All vote}

Let us now consider the case in which $V=C$. Recall that $B(\sigma)$ is the set of customers bribed by $\sigma$.
We say that two strategies $\sigma_1$ and $\sigma_2$ are \emph{disjoint} if $B(\sigma_1)\cup B(\sigma_2)=\emptyset$.
By direct calculation it follows that the revenue of disjoint strategies exhibits the following property:

\begin{lemma}\label{lem:linear}
If $V=C$ and $\sigma_1$ and $\sigma_2$ are two disjoint strategies, then $\rev_\O(\sigma_1 \circ \sigma_2)=\rev_\O(\sigma_1)+\rev_\O(\sigma_2)$.
\end{lemma}


\noindent
We now show that bribing a single individual is not profitable.

\begin{lemma}\label{lem:atomic}
Let $\sigma$ be a bribing strategy, $V=C$ and $|B(\sigma)|=1$. Then, $\rev_\O(\sigma) \leq 0$, i.e., $\sigma$ is not profitable.
\end{lemma}

\begin{proof}[Proof sketch]
Let $\bar c$ be the only individual such that $\sigma(\bar c)\not = 0$.
By calculation, $\rev(\sigma) = u_\O^\sigma - u_\O^0= 
\orat^\sigma - \orat -\sum_c \sigma(c)  =
\min \{1, \eval(\bar c)+\sigma(\bar c) \} - \eval (\bar c) - \sigma (\bar c) \leq 0$. 
\end{proof}

\noindent
By combining the two lemmas above we are able to show that no strategy is profitable for bribing the $\orat$.

\begin{proposition}
If $V=C$, then no strategy is profitable.
\end{proposition}

\begin{proof}[Proof sketch]
Any bribing strategy $\sigma$ can be decomposed into $n$ pairwise disjoint strategies such that $\sigma=\sigma_{c_1}\circ\dots \circ \sigma_{c_n}$ and $|B(\sigma_{c_j})|=1$ for all $1\leq j\leq n$.
By applying Lemma~\ref{lem:linear} and Lemma~\ref{lem:atomic} we then obtain that $\rev_\O(\sigma)\leq 0$. 
\end{proof}

\noindent
From this it follows that $\sigma^{0}$ is weakly dominant and thus $\orat$ bribery-proof when all customers voted.

%
%
%
%
%
%
%

\subsection{Non-voters}

Let us now consider the case of $V\subset C$, i.e., when there is at least one customer who is not a voter. 
In this case Lemma~\ref{lem:linear} no longer holds, as shown in the following example.

\begin{example}\label{example:non-voters}
Let $C=\{A, B, C\}$, and let $\eval(A)= 0.5$, $\eval(B)=0.5$, and $\eval(C)=*$. 
The initial resources are $u^0=\orat \times 3=1.5$.
Let now $\sigma_1(A)=0.5$ and $\sigma_1(B)=\sigma_1(C)=0$, and let $\sigma_2(C)=0.5$ and $\sigma_2(A)=\sigma_2(B)=0$. 
Now $u_\O^{\sigma_1}=0.75\times 3-0.5=1.75$ and $u_\O^{\sigma_2}=0.5\times 3-0.5=1$, but $u_\O^{\sigma_1\circ \sigma_2}=0.\bar{6} \times 3 -1 = 1$.
\end{example}

The example (in particular $\sigma_1$) also shows that $\orat$ in this case is not bribery-proof. 

We now turn to characterise the set of undominated bribing strategies. We begin by showing that bribing a non-voter is always dominated. 
Let first ${\sigma}$ be a strategy such that \mbox{${\sigma}(\bar c) \neq 0 $} for some $\overline{c}\in C \setminus V$ and recall that $V^\sigma$ is the set of voters after execution of $\sigma$.
Let us define the $\overline{c}$-{\em greedy restriction} of ${\sigma}$ to be any strategy $\sigma^{-\bar c}$ such that:

\begin{itemize}

\item $V^{\sigma^{-\bar c} } = V^{\sigma} \setminus \{\overline{c}\}$, i.e., the greedy restriction eliminates $\bar c$ from the set of voters.

\item For each $c \in V^{\sigma} \setminus \overline{c}$, $\max(1, eval(c)+ {\sigma}(c)) = \max(1, eval(c)+ {\sigma^{-\bar c}}(c))$, i.e.,  the greedy restriction does not waste further resources.

\item If there exists $c \in V^{\sigma} \setminus \overline{c}$ such that $eval(c)+ {\sigma^{-\bar c}}(c)<1$ then $\sum_{c \in C}\sigma^{-\bar c}(c)=\sum_{c \in C}\sigma(c)$, i.e., the $\sigma^{-\bar c}$ redistributes $\sigma(\bar c)$ among the remaining voters.

\end{itemize}

\noindent
We now show that each strategy bribing a non-voter is strictly dominated by any of its greedy restrictions.

\begin{proposition}\label{prop:non-voters}
Let $V\not =C$, and $\bar{c}\in C\setminus V$. 
Then each strategy $\sigma$ with $\sigma(\overline{c}) \neq 0$ is strictly dominated by $\sigma^{-\bar c}$.
\end{proposition}

\begin{proof}
Let $\sigma$ be a strategy with $\sigma(\overline{c}) \neq 0$ for some non-voter $\overline{c}$, and let $\sigma^{-\bar c}$ be one of its greedy restriction defined above.
\begin{small}
\begin{eqnarray*}
u^{\sigma^{-\bar c}}_\O-u^{\sigma}_\O = \\
  |C|(\orat^{\sigma^{-\bar c}} {-} \orat^{\sigma} ) {+} \sum_{c \in C}\sigma(c)
   {-} \sum_{c \in C} \sigma^{-\bar c}(c) = \\
   |C|(\frac{\sum_{c \in C} \eval^{\sigma^{-\bar c}}(c)}{|V|} - \frac{\sum_{c \in C} eval^\sigma(c)}{|V \cup \overline{c}|})+\\
    + (\sum_{c \in C}\sigma(c) - \sum_{c \in C}\sigma^{-\bar c}(c))
    \end{eqnarray*}
\end{small}
\noindent
Observe first that $\sigma^{-\bar c}$ is a redistribution, hence $\sum_{c}\sigma(c) - \sum_{c}\sigma^{-\bar c}(c)\geq 0$, i.e., the second addendum in the above equation is positive.
Consider now the case where there exists $c \in V^{\sigma} \setminus \overline{c}$ such that $eval(c)+ \sigma^{-\bar c}(c)<1$. Then by the definition of $\sigma^{-\bar c}$ we have that $\sum_{c \in V^{\sigma}} eval^{\sigma}(c) = \sum_{c \in V^{\sigma^{-\bar c}}} eval^{\sigma^{-\bar c}}(c) $, i.e., the greedy restriction preserves the overall evaluation.
%
By straightforward calculation this entails that $u^{\sigma^{-\bar c}}_\O-u^{\sigma}_\O>0$.
If no such $c$ exists, and therefore $\orat^{\sigma^{-\bar c}}=1$ we have that either $\orat^{\sigma}<1$ or, by the efficiency requirement and the fact that  $\sigma(\overline{c})\neq 0$, we have that $ \sum_{c \in C}\sigma(c)> \sum_{c \in C}\sigma^{-\bar c}(c)$.
%
In either cases we have that $u^{\sigma^{-\bar c}}_\O-u^{\sigma}_\O>0$.
\end{proof}


Let an  {\em $\O$-greedy strategy} be any efficient strategy that redistributes all the initial resources $u^0_\O$ among voters. Making use of the previous result, we are able to characterise the set of all dominant strategies for $\orat$.

\begin{proposition}\label{prop:orating-optimal}
Let $V\not =C$. A strategy is weakly dominant for $\orat$ if and only if it is an $\O$-greedy strategy.
\end{proposition}

\begin{proof}[Proof sketch.]
For the right-to-left direction, first observe that all $\O$-greedy strategies are payoff-equivalent, and that a non-efficient strategy is always dominated by its efficient counterpart. 
By Proposition~\ref{prop:non-voters} we know that strategies bribing non-voters are dominated, and by straightforward calculations we obtain that in presence of non-voters it is always profitable to bribe as much as possible.
For the left-to-right direction, observe that a non-greedy strategy is either inefficient, or it bribes a non-voter, or does not bribe as much as possible. In either circumstance it is strictly dominated.
\end{proof}

While there may be cases in which the number of weakly dominant strategies under $\orat$ is exponential, all such strategies are revenue equivalent, and Proposition~\ref{prop:orating-optimal} gives us a polynomial algorithm to find one of them: starting from an evaluation vector $\eval$, distribute all available resources $u^0_\O$ to the voters, without exceeding the maximal evaluation of 1. By either exhausting the available budget or distributing it all, we are guaranteed the maximum gain by Proposition~\ref{prop:orating-optimal}.

\section{Bribes under $\prat$}\label{sec:personalised}


 

In this section we look at bribing strategies under $\prat$, against various knowledge conditions on the social network.
As for Section~\ref{sec:objective} we start by looking at the case where everyone votes and later on allowing non-voters. Before doing that, we introduce a useful graph-theoretic measure of influence.

\begin{definition}\label{def:weights}
The \emph{influence weight} of a customer $c\in C$ in a network $E$ and and a set of designed voters $V$ is defined as follows:
$$w_c^V=\sum_{k\in N(c)} \frac{1}{|N(k) \cap V|}$$
\end{definition}

\noindent
Recall that we assumed that every customer can see a voter, thus $w_c^V$ are well-defined for every $c$.
If $V=C$, i.e., when everybody voted, we let $w_c=w_c^C$. In this case, we obtain $w_c=\sum_{k\in N(c)} \frac{1}{\deg(k)}$, where $\deg(c)=|N(c)|$ is the \emph{degree} of $c$ in $E$.
When $V$ is defined by a bribing strategy $\sigma$, we write $w^\sigma_c=w^{V^\sigma}_c$.

Intuitively, each individual's rating influences the rating of each of its connections, with a factor that is inversely proportional to the number of second-level connections that have expressed an evaluation. We formalise this statement in the following lemma:

\begin{lemma}\label{lem:static}
The utility obtained by playing $\sigma$ with $\prat$ is 
$u^\sigma_\P=\sum_{c\in V^{\sigma}} w_c^\sigma \times \eval^\sigma(c) - \sum_{c \in C} \sigma(c)$.
\end{lemma}

\begin{proof}[Proof]
By calculation:
\begin{eqnarray*}
u^\sigma_\P + \sum_{c\in  C}\sigma(c) = \sum_{c\in C}\prat^\sigma(c)= \sum_{c\in C} \avg_{k\in N(c) \cap V^{\sigma}} \eval^\sigma(k)=\\
=\sum_{c\in C} \big [\frac{1}{|N(c)\cap V^{\sigma}|} \sum_{k\in N(c)\cap V^{\sigma}} \eval^\sigma(k) \big ]=\\
=\sum_{k\in V^{\sigma}} \big [ \; \eval^\sigma(k) \times \sum_{k'\in N(k)} \frac{1}{|N(k') \cap V^{\sigma}|} \big ]=\\
=\sum_{c\in V^{\sigma}} w_c^\sigma \times \eval^\sigma(c)
\end{eqnarray*}

\vspace{-0,4cm}
\end{proof}

\subsection{All vote, known network}

We begin by studying the simplest case in which the restaurant knows the evaluation $\eval$, the network $E$ as well as the position of each customer on the network.
The following corollary is a straightforward consequence of Lemma~\ref{lem:static}:

\begin{corollary}\label{cor:linear}
Let $V=C$ and let $\sigma_1$ and $\sigma_2$ be two disjoint strategies, then $\rev_\P(\sigma_1 \circ \sigma_2)=\rev_\P(\sigma_1)+\rev_\P(\sigma_2)$.
\end{corollary}

\noindent
We are now able to show a precise characterisation of the revenue obtained by any efficient strategy $\sigma$: 

\begin{proposition}\label{prop:prating-voters}
Let $V=C$, let $E$ be a known network, and let  $\sigma$ be an efficient strategy. 
Then 
$\rev_\P(\sigma) = \sum_{c\in C} (w_c-1) \sigma(c)$.
\end{proposition}

\begin{proof}
By calculation, where Step (2) uses Lemma~\ref{lem:static}, 
and Step (4) uses the fact that $\sigma$ is efficient:

\begin{eqnarray}
\rev_\P(\sigma) = u^\sigma_\P - u^0_\P =\\
= [\sum_{c\in C} w_c \; \eval^\sigma_c  - \sum_{c\in C} \sigma(c) - \sum_{c\in C} w_c\; \eval(c)] =\\
= \sum_{c\in C} \big [ w_c\; [\min\{1,\eval(c)+\sigma(c)\} - \eval(c)] \big] 
- \sum_{c\in C}\sigma(c) \\= \sum_{c\in C} (w_c -1) \sigma(c).
\end{eqnarray}

\vspace{-0.4cm}

\end{proof}

Proposition~\ref{prop:prating-voters} tells us that the factors $w_c$ are crucial in determining the revenue of a given bribing strategy.
Bribing a customer $c$ is profitable whenever $w_c{>}1$ (provided its evaluation was not $1$ already), while bribing a customer $c$ with $w_c{\leq}1$  is at most as profitable as doing nothing, as can be seen in the example below.
Most importantly, it shows that $\prat$ is {\em not} bribery-proof when the restaurant knows both the network and the customers' evaluations. 

\begin{example}
Let $E$ be a four arms stars, and let $A$ be the individual in the centre. 
Assume each individual values the restaurant 0.5. We have that $w_A=2.2$ and $w_c=0.7$ for all $c$ different from $A$. 
Consider now two bribing strategies: $\sigma^A$ which bribes $A$ with $0.5$, and $\sigma^B$ which bribes a single individual $B\not= A$ with the same amount. 
What we obtain is that $\rev_\P(\sigma^A)=0.6$, while $\rev_\P(\sigma^B)=-0.15$.
\end{example}

Given a network $E$ and an evaluation vector $\eval$, let Algorithm \ref{algo:Pgreedy} define the \emph{$\P$-greedy bribing strategy}.


\begin{algorithm}[ht]
{
 \KwIn{Evaluation function $\eval$ and network $E$}
 \KwOut{A bribing strategy $\sigma^G_\P : C\to \text{\it Val}$}

\smallskip

\textit{Budget}=$u_\P^0$\\
$\sigma_\P^G(c)=0$ for all $c\in C$\\
Compute $w_c$ for all $c\in C$\\
Sort $c\in C$  in descending order $c_0,\dots, c_m$ based on $w_c$\\

\smallskip

\For{i=0,\dots,m} {

\If{ \textit{Budget}$\not = 0$} {
\If { $w_{c_i} >1$ } {

$\sigma^G_\P(c_i)=\min\{1-\eval(c_i), \text{\it Budget}\}$\\

\textit{Budget}=\textit{Budget}-$\sigma^G_\P(c_i)$
}
}

\smallskip

\Return $\sigma^G_\P$\
}
\medskip}
\caption{The $\P$-greedy bribing strategy $\sigma^G_\P$}
\label{algo:Pgreedy}
 \end{algorithm}

As a consequence of Proposition~\ref{prop:prating-voters} we obtain:

\begin{corollary}\label{prop:prating-greedy}
The $\P$-greedy bribing strategy defined in Algorithm~\ref{algo:Pgreedy}
is weakly dominant.
\end{corollary}


As in the case of $\orat$, Corollary~\ref{prop:prating-greedy} has repercussions on the computational complexity of bribery: it shows that computing a weakly dominant strategy can be done in polynomial time. Notice how the most costly operation lies in the computation of the influence weights $w_c$, which can be performed only once, assuming the network is static.
Similar problems, such as recognising whether bribing a certain individual is profitable, or estimating whether individuals on a network can be bribed above a certain threshold, are also computable in polynomial time.

\subsection{All vote, unknown network}

We now move to study the more complex case of an unknown network. 
Surprisingly, we are able to show that no bribing strategy is profitable (in expectation), and hence $\prat$ is bribery-proof in this case.
Recall that we are still assuming that the restaurant knows $\eval$ and everybody voted.

We begin by assuming that the restaurant knows the structure of the network, but not the position of each participant. Formally, the restaurant knows $E$, but considers any permutation of the customers in $C$ over $E$ as possible.
Let us thus define the expected revenue of a strategy $\sigma$ over a given network $E$ as the average over all possible permutations of customers:
$\mathbb E[\rev_\P(\sigma)] = \sum_{} \frac{1}{n!} [u_\rho^\sigma - u_\rho^0],$
where we abuse notation by writing $u_\rho^\sigma$ as $u_\P^\sigma$ under permutation $\rho$ over the network $E$.
What we are able to show is that all strategies are at most as profitable as $\sigma^0$ in expected return:

\begin{proposition}\label{prop:prating-uncertain}
Let $V=C$, let the network structure of $E$ be known but not the relative positions of customers on $E$. Then $\mathbb E[\rev_\P(\sigma)]=0$ for all strategies $\sigma$.
\end{proposition}

\begin{proof}[Proof sketch.]
Let $|C|=n$. 
We show the result for any strategy $\sigma$ that bribes a single customer $\bar c$. 
The general statement follows from the linearity of $\mathbb E[\rev(\sigma)]$. Equation (5) uses Proposition~\ref{prop:prating-voters} to compute the revenue for each permutation~$\rho$ of customers $C$ on the network:
\begin{small}
\begin{eqnarray}
\mathbb E [\sigma] = \sum_{\rho} \frac{1}{n!}\;\;  (u^{\sigma}_{\rho} - u^0_\rho)  
=\sum_{\rho} \frac{1}{n!} (w_{\rho(\bar c)} - 1)\sigma(\bar c) = \\
=\sum_{c\in C} \frac{(n-1)!}{n!} (w_c-1) \sigma(\bar c) = \frac{(n-1)!}{n!} \sum_{c\in C} (w_c -1) = 0
\end{eqnarray}
\end{small}
The last line follows from the observation that $\sum_c w_c = |C|$ and hence $\sum_{c} (w_c -1) = 0$, by a consequence of Definition~\ref{def:weights} when everybody votes. 
\end{proof}

Hence, if we assume a uniform probability over all permutations of customers on the network, a straightforward consequence of Proposition~\ref{prop:prating-uncertain} concludes that it is not profitable (in expectation) to bribe customers.

\begin{corollary}
If $V=C$ and the network is unknown, then no bribing strategy for $\prat$ is profitable in expected return.
\end{corollary}

\subsection{Non-voters, known network}

With $\prat$ it is possible to find a network where bribing a non-voter is profitable:


\begin{example}\label{example:star}
Consider 4 individuals $\{B,C,D,E\}$ connected only to a non-voter in the middle. 
Let $eval(j)=0.2$ for all $j$ but the center. 
We have $u^0_\P=1$. 
Let $A$ be the non-voter, and let $\sigma_1(A)=1$ and 0 otherwise. 
The utility of $\sigma_1$ is:
$$\prat^{\sigma_1}(A)+4 \prat^{\sigma_1}(j) -1 = 1.76$$

\noindent
All other strategies can be shown to be dominated by $\sigma_1$. 
Take for instance a strategy $\sigma_2$ such that $\sigma_2(B)=0.8$, $\sigma_2(C)=0.2$ and 0 otherwise. The utility of $\sigma_2$ is $u^{\sigma_2}_\P= 1.25$.
\end{example}

It is quite hard to obtain analytical results for strategies bribing non-voters, due to the non-linearity of the $\prat$ in this setting.  
We can however provide results in line with those of the previous section if we restrict to \emph{voter-only strategies}, i.e., strategies $\sigma$ such that $\sigma(c)=0$ for all $c\not\in V$.
In this case, a similar proof to Proposition~\ref{prop:prating-voters} shows the following:

\begin{proposition}
Let  $V\not = C$, $E$ be a known network, and $\sigma$ be an efficient bribing strategy such that $B(\sigma)\subseteq V$.
Then,
$\rev_\P(\sigma) = \sum_{c\in V} (w^V_c-1) \sigma(c)$.
\end{proposition}

\noindent
The difference with the case of $V=C$ is that $w_c^V$ can be arbitrarily large in the presence of non-voters, such as in our Example~\ref{example:star}.


\subsection{Non-voters, unknown positions}

Unlike the case of $V=C$, in this case it is possible to define bribing strategies that are profitable (in expected return). 

\begin{example}\label{ex:triangle}
Let $C=\{A,B,C\}$, and the initial evaluation $\eval(A)=\eval(B)=0.2$ and $\eval(C)=*$. 
Assume that the structure of the network is known, but the position of the individuals is not. Let the three possible network positions (without counting the symmetries) be depicted in Figure~\ref{fig:triangle}. Let $\sigma(B)=0.2$ and $\sigma(A)=\sigma(C)=0$. In the first case:
{\small
\begin{eqnarray*}
\rev^1_P(\sigma){=}\prat(A)+...+\prat(C) - 0.2 - u_\P^0 =\\
=0.3 + 0.3 + 0.4 - 0.2 - 0.6 = 0.2
\end{eqnarray*}
}
\noindent
In the second case $\rev^2_\P(\sigma)=0$ while in the third:
\begin{eqnarray*}
\rev^3_P(\sigma)=0.4 +0.3 +0.2 - 0.2 -0.6 = 0.1
\end{eqnarray*}
\end{example}

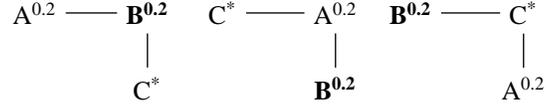
\begin{figure}[t]
\begin{center}

\begin{tikzpicture}[auto]
  \node          (a)  at (-2.5,0)  {A\textsuperscript{0.2}};
  \node          (b) at (-1,0)   {\textbf{B\textsuperscript{0.2}}};
  \node          (c) at (-1,-1) {C\textsuperscript{*}};

  \path[->]    (a) edge [-, line width=0.5pt]        node       {$$} (b)
			   (b) edge [-, line width=0.5pt]        node       {$$} (c);

  \node          (a1)  at (0,0)  {C\textsuperscript{*}};
  \node          (b1) at (1.5,0)   {A\textsuperscript{0.2}};
  \node          (c1) at (1.5,-1) {\textbf{B\textsuperscript{0.2}}};

  \path[->]    (a1) edge [-, line width=0.5pt]        node       {$$} (b1)
			   (b1) edge [-, line width=0.5pt]        node       {$$} (c1);

  \node          (a1)  at (2.5,0)  {\textbf{B\textsuperscript{0.2}}};
  \node          (b1) at (4,0)   {C\textsuperscript{*}};
  \node          (c1) at (4,-1) {A\textsuperscript{0.2}};

  \path[->]    (a1) edge [-, line width=0.5pt]        node       {$$} (b1)
			   (b1) edge [-, line width=0.5pt]        node       {$$} (c1);

\end{tikzpicture}
\caption{Customers permutations in Example~\ref{ex:triangle}.}\label{fig:triangle}

\end{center}
\end{figure}

Therefore, $\prat$ is not bribery-proof (in expectation) in the presence of non-voters when the network is unknown.
Interesting computational problems open up in this setting, such as identifying the networks that allow for profitable bribing strategies, and their expected revenue.

\section{Boundaries of bribery-proofness}\label{sec:strategyproof}

The previous sections have shown that having a network-based rating systems, where individuals are influenced by their peers, is not bribery-proof, even when the position of individuals in a given network is not known.
However bribing strategies have a different effect in the overall score. While the utility of $\orat$  is a sum of the {\em global} average of voters' evaluation, the utility of \prat is a sum of {\em local} averages of voters' evalution against the one of their peers.

Therefore a strategy bribing one voter affects everyone in the case of \orat, but it can be shown to have a limited effect in the case of \prat.

\begin{proposition}
Let $\sigma$ be an efficient strategy s.t. $|B(\sigma)|=1$, and let $\bar c$ be 
such that $\sigma(c)\not = 0$.
Then $\rev_\P(\sigma)<N(\bar c)$.

\end{proposition}

\begin{proof}
By calculation, we have that:
\begin{small} 
\begin{eqnarray*}
\rev_\P(\sigma) = 
\sum_{c\in C}  \prat^{\sigma}(c) {-} \sigma(\bar c) {-} \sum_{c\in C}  \prat(c)  = \\ 
\sum_{c' \in N(\bar c)} \prat^{\sigma}(c)  - \sigma(\bar c) - \sum_{c' \in N(\bar c)} \prat(c)  \leq \\
\leq 1 \times N(\bar c) - \sigma(\bar c) - \sum_{c' \in N(\bar c)} \prat(c) < N(\bar c)
\end{eqnarray*}
\vspace{-0.6cm}
\end{small}

\end{proof}

The previous result shows that increasing the number of individuals that are not connected to an agent that is bribed, even if these are non-voters, does not increase the revenue of the bribing strategy. 
This is not true when we use $\orat$.

\begin{proposition}
Let $\sigma$ be an efficient strategy.
The revenue $\rev_\O(\sigma)$ of $\sigma$ is monotonically increasing with the number of non-voters, and is unbounded.
\end{proposition}

\begin{proof}
It follows from our definitions that:
$$ \rev_\O(\sigma)=  (\frac{|C|}{|V^\sigma|} - 1) \big [ \sum_{c\in C} \eval(c)+\sigma(c) \big ]$$
The above figure is unbounded and monotonically increasing in the number of non-voters, which can be obtained by increasing $C$ keeping $V^\sigma$ fixed.
\end{proof}

So while $\prat$ and $\orat$ are not bribery-proof in general, it turns out that the impact of the two in the overall network are significantly different. 
In particular, under realistic assumptions such as a very large proportion of non-voters and with participants having a few connections, bribing under $\orat$ is increasingly rewarding, while under $\prat$ this is no longer the case.

\section{Conclusive remarks}\label{sec:conclusions}


We introduced $\prat$, a network-based rating system which generalises the commonly used $\orat$, and analysed their resistance to external bribery under various
 conditions.
The main take-home message of our contribution can be summarised in one point, deriving from our main results:

$\prat$ and $\orat$ are not bribery-proof in general. 
However, if we assume that a service provider has a cost for bribing an individual, there are situations in which $\prat$ is fully bribery proof, while $\orat$ is not.
For instance, if the cost of bribing an individual $c$ is at least $N(c)$ then $\prat$ is bribery-proof.
As observed previously, this is not necessarily true for $\orat$. In particular, if we assume the presence of unreachable individuals the difference is more significant. As shown, for $\prat$ we need to bribe individuals with $w_c>1$. 
With $\orat$ is sufficient to find one voter who accepts a bribe.


There is a number of avenues open to future research investigation. 
The most important ones include the case of partially known customers' evaluation, and the study of ratings of multiple restaurants, where the probability of a customer choosing a restaurant determines his or her probability not to choose the others. 

\section*{Acknowledgments}
Umberto Grandi acknowledges the support of the Labex CIMI project ``Social Choice on Networks'' (ANR-11-LABX-0040-CIMI). Paolo Turrini the support of Imperial College London for the Junior Research Fellowship ”Designing negotiation spaces for collective decision-making” (DoC-AI1048). This work has also been partly supported by COST Action IC1205 on Computational Social Choice. 

\bibliographystyle{named}
\bibliography{pr}

\end{document}